\documentclass[12pt]{article}

\usepackage[margin=2.5cm]{geometry}

\usepackage{amsmath,amsthm,amsfonts,amscd,amssymb,mathrsfs,bbm,enumerate,url}

\usepackage{graphicx,tikz}
\usepackage[pdfborder={0 0 0}]{hyperref}
\usetikzlibrary{matrix,arrows}
\usetikzlibrary{decorations.pathreplacing}
\usetikzlibrary{decorations.pathmorphing}
\usetikzlibrary{patterns,fadings}
\usepackage{soul}
\usepackage{xcolor}

\def\RR{{\mathbb R}}
\def\CC{{\mathbb C}}
\def\NN{{\mathbb N}}
\def\ZZ{{\mathbb Z}}

\def\A{{\mathcal A}}

\def\cH{{\mathcal H}}

\def\M{{\mathcal M}}

\def\b{\beta}

\def\g{\gamma}

\def\k{\kappa}

\def\l{\lambda}
\def\cW{{\cal W}}

\def\L{\Lambda}

\def\dim{{\hbox{dim}\,}}

\def\1{{\mathbbm 1}}

\def\uone{{\rm U(1)}}

\def\diff{{\rm Diff}}

\def\diffs1{\diff(S^1)}
\def\mob{{\rm M\ddot{o}b}}

\def\supp{{\rm supp\,}}

\def\psl2r{{\rm PSL}(2,\RR)}
\def\sl2r{{\rm SL}(2,\RR)}
\def\su11{{\rm SU}(1,1)}
\def\2dmob{{\overline{\psl2r}\times\overline{\psl2r}}}
\def\<{\langle}
\def\>{\rangle}

\def\dom{{\mathrm{Dom}}}

\newtheorem{theorem}{Theorem}[section]

\newtheorem{corollary}[theorem]{Corollary}
\newtheorem{proposition}[theorem]{Proposition}
\newtheorem{lemma}[theorem]{Lemma}
\newtheorem{conjecture}[theorem]{Conjecture}
\theoremstyle{remark}
\newtheorem{remark}[theorem]{Remark}

\title{Local energy bounds and strong locality in chiral CFT}
\date{} 
\author{
{\bf Sebastiano Carpi},
{\bf Yoh Tanimoto}
\\
   Dipartimento di Matematica, Universit\`a di Roma Tor Vergata\\
   Via della Ricerca Scientifica 1, I-00133 Roma, Italy\\
email: {\tt carpi@mat.uniroma2.it}, {\tt hoyt@mat.uniroma2.it}\\
\\
{\bf Mih\'aly Weiner}\\
Budapest University of Technology and Economics (BME), \\
Department of Analysis, \\
  H-1111 Budapest  M\H{u}egyetem rkp. 3--9 Hungary, and\\
  MTA-BME Lend\"ulet Quantum Information Theory Research Group\\
email: {\tt mweiner@math.bme.hu}
}
\begin{document}
\maketitle

\begin{abstract}
A family of quantum fields is said to be strongly local if it generates a local net of von Neumann algebras.
There are few methods of showing \textit{directly} strong locality of a quantum field. Among them, linear energy bounds are the most widely used, yet a chiral conformal field of conformal weight $d>2$ cannot admit linear energy bounds.
In this paper we give a new direct method to prove strong locality in two-dimensional conformal field theory. We prove that if a chiral conformal field satisfies an energy bound of degree $d-1$, then it also satisfies a certain local version of the energy bound, and this in turn implies strong locality. A central role in our proof is played by diffeomorphism symmetry.
As a concrete application, we show that the vertex operator algebra given by a unitary vacuum representation of the $\cW_3$-algebra is strongly local. For central charge $c > 2$,  this yields a new conformal net.
We further prove that these nets do not satisfy strong additivity, and hence are not completely rational.
\end{abstract}

\section{Introduction}\label{intro}

Locality plays a fundamental role in quantum field theory.
Its precise formulation depends on setting, 
but it is usually expressed as some kind of
commutation  relation. In the Haag-Kastler framework \cite{Haag96}, we have commutativity of the bounded operators associated to spacelike separated regions. 
When the theory is formulated in terms of {\it quantum fields} (operator-valued distributions), one requires
commutativity between the operators obtained after smearing the fields with test functions having space-like separated supports (Wightman locality \cite{SW00}).
However, unlike in the previous case, these are -- in general -- unbounded operators whose domain is not the full Hilbert space, so what is in fact assumed is commutativity on some (fixed, suitable) domain. In the chiral conformal case, one often uses the setting of {\it vertex operator algebras} (VOAs) in which locality appears as a certain commutation relation between {\it formal series} \cite{Kac98} which is directly related to Wightman locality, see e.g.\! \cite[Section 1.1]{Kac98} and \cite[Appendix A]{CKLW18}.

Conceptually, these different forms of locality should all express the same physical principle. However, their mathematical equivalence is far from being evident. Nelson's famous example \cite[Section 10]{Nelson59} shows that, even if two self-adjoint operators $A,B$ commute on a common dense core, the (bounded) unitary operators $e^{iAt}$ and $e^{iBs}$ (where $t,s\in\RR$) need not commute. The commutativity of these 
bounded functions is a {\it stronger} property, which is hence called {\it strong commutativity}. The difference between 
commutativity and strong commutativity is a major obstacle for moving from (unbounded) quantum fields to Haag-Kastler nets of (bounded) algebras.

In prominent examples of constructive quantum field theory \cite{GJ87}, this problem has been treated by linear energy bounds: when a suitable positive self-adjoint operator (``Hamiltonian'') bounds the smeared fields and their commutators, then strong locality follows. However, there are important quantum fields which \textit{do not} satisfy linear energy bounds. 

In this paper we focus on chiral conformal quantum field theory.
Apart from physical reasons, their study is also well-motivated by mathematical ones. There is a rich variety of explicitly constructed chiral conformal models and they helped to uncover deep mathematical relations between various different areas of mathematics.
Here we consider the question of under what condition one can guarantee that two ``formally'' relatively local chiral conformal fields -- i.e.\! ones that are relatively local as formal  series -- remain local even at the level of generated bounded (von Neumann) algebras. In case they do so, we say that they are relatively {\it strongly local}. One of our main motivations is the study of connection of between VOAs and conformal nets which, from a general point of view, was initiated in \cite{CKLW18}, a work in which the notion of strong locality of fields plays a central role.

The passage from formal series (VOA setting) to algebras of bounded operators on a Hilbert space (i.e.\! to conformal Haag-Kastler nets, or conformal nets for short) requires three ingredients. Firstly, one needs the existence of a suitable inner product -- as it is usually put, the model must be {\it unitary}. This is far from being an easy problem; for example, for the Virasoro algebra, there is a parameter $c$ called the {\it central charge} and it was a great achievement \cite{GKO86, FQS86} to establish the discrete series and the continuous family of $c$ for which there are unitary representations. Secondly, even having a suitable inner product, one needs to have some control on the ``unboundedness'' of the Fourier modes of the fields so that smearing with 
smooth test functions
can be defined. This is usually done by establishing and using some {\it polynomial energy bounds}, see e.g.\! at \cite[Section 6]{CKLW18}. Contrary to the question of unitarity, this is often easy: we are unaware of any unitary model, where such bounds could not be derived by rather standard, straightforward arguments for the generating fields (which then implies polynomial bounds for all fields).
Lastly, one has to prove that the smeared fields strongly commute when the test functions have disjoint support.

In the context of unitary VOAs, solutions to the question of strong locality may be classified into two types.
There are some methods -- although rather few; most prominently the one relying on the use of linear energy bounds (which we will also review below in more detail) -- with which one can directly show the strong relative locality of two fields. These may be called {\it primary} methods. On the other hand, the strong relative locality between fields of a generating set implies 
the strong relative locality of any two fields of the VOA
\cite{CKLW18}.
This gives a {\it secondary} way of showing strong locality:  one can prove strong locality of a VOA by showing that it can be embedded in a larger VOA which is generated by fields whose strong locality can be proven by a \textit{primary method}. Other secondary methods concerns the VOA representation theory and give the strong locality for some unitary VOA extensions of certain strongly local VOAs, see e.g.\!
\cite[Section 5]{Gui20-1}.

A conformal field $\phi$ is said to satisfy a \textit{linear energy bound} if 
for every test function $f$, the product
\[
 \phi(f) (L_0 + \1)^{-1},
\]
(where $L_0$ is the conformal Hamiltonian) is bounded. 
Existence of such bounds allows an application of the Glimm-Jaffe-Nelson commutator theorem \cite{Nelson72, GJ72} showing the desired strong locality \cite{DF77}. As famous examples for fields satisfying such a bound, we mention the stress-energy field and
the free bosonic field (or as is also called: the $U(1)$-current). The free fermionic field satisfies an even better bound:
the smeared field is already bounded, without need of a  dumping by $(L_0+\1)^{-1}$, hence strong locality of the (even part) of the relevant VOA is immediate.  
Yet, there are interesting unitary VOAs that are not generated by fields satisfying linear energy bounds and have no known realizations as sub-VOAs of some unitary VOA whose strong locality is already established (by e.g.\! the use of linear energy bounds). In fact, in some sense, the existence of linear energy bounds is rather rare: as we shall show in Section \ref{sec:globaltolocal}, a (nonzero) primary field $\phi$ of conformal dimension $d>1$ 
can \textit{at best}  satisfy an energy bound of degree $d-1$; that is
\[
 \phi(f) (L_0 + I)^{-s}
\]
may be bounded for all test functions $f$ only
if $s\geq d-1$ (with the $s=d-1$ case hence referred to  as an \textit{optimal} energy bound). Thus, a primary field of conformal dimension $d>2$ cannot satisfy a linear energy bound. 

An optimal energy bound, even if it is not a linear one, can be of great value. One of our main results is the following: if $\phi$ is a primary field of conformal dimension $d>1$ and it admits an optimal energy bound, then in turn the following also holds: for any nonnegative test function $f\geq 0$, 
\[
\phi(f^{d-1}) (T(f) + (r(f)+\epsilon)I)^{-(d-1)}
\]
is a bounded operator, where $T(f) = \sum_n \hat f_n L_n$ is the smeared stress-energy tensor
(which is bounded from below by \cite{FH05})
and sufficiently large positive number $r(f)$.
This is a {\it local} version of the energy bounds: rather than using $L_0$ -- which is $T$ smeared with the constant $1$ function -- we have a bounding positive operator whose localization coincides with that of $\phi(f)$. 

The idea of local energy bounds has already circulated in the last few years and such bounds have been
shown for fields (currents) with $d=1$ \cite{CW-localenergy},
but so far no such bounds were derived for the $d > 2$ case. Local energy bounds are interesting by themselves, because they pass immediately to representations. Besides representation theory, local energy bounds were used in the construction of spectral triples in the 
superconformal net setting \cite{CHKL10, CHKLX15}.
Here we shall use them in the 
context of strong locality. In particular, we prove that if $\phi,\tilde{\phi}$ are relatively local primary fields of conformal dimensions $d,\tilde{d} > 1$ satisfying the above local energy bound, then they are also strongly relatively local. Thus, an optimal (global) energy bound implies a local one which in turn can be used to derive strong locality. Altogether, this gives a new primary method for showing strong locality.

As a concrete example of all these new techniques, we consider the $\mathcal W_3$-algebra. Recently we showed that its vacuum representations are unitary for every central charge value $c\geq 2$ \cite{CTW22}. This model has a certain realization in the tensor product of two $U(1)$-Heisenberg algebras, but for $c>2$ this realization does not respect the native Hamiltonian and inner product of the Heisenberg algebra and hence cannot be directly used to conclude strong locality. 
Having no further known realizations for the general $c>2$ case, this is a perfect place for an application of our new method. 

By rather involved computations (including the estimate of the normal square $:\!L^2\!:$ of the Virasoro field), we check that the $W$ field (of conformal weight $d=3$) satisfies an optimal global energy bound of degree $3-1 = 2$.
This implies that the $W$ field satisfies a local energy bound, and hence -- as this field together with the stress-energy tensor generate the full model -- strong locality of the $\mathcal W_{3,c}$ unitary VOA for all $c\ge 2$.
Besides the applications to strong locality of VOAs the results of this paper should also provide a new useful tool in the comparison between 
vertex operator algebras and conformal nets from the representation theory point of view which has been considerably developed in the last few years, see e.g.\! \cite{Gui20-2}, \cite{Tener19-2}.

This paper is organized as follows. In Section \ref{preliminaries} we recall fundamental concepts,
in Section \ref{sec:main} we present our main results: how an optimal global energy bound implies a local one in Section \ref{sec:globaltolocal}, and how these local energy bounds imply strong locality in Section \ref{sec:strong}.
In Section \ref{improving} we tools to obtain global energy bounds from bounds on commutators.
In Section \ref{sec:W} we study the particular case of the $\mathcal W_3$-algebra.
By the techniques developed above, we conclude that the $W$ field is strongly local to itself for $c\ge 2$. We conclude the paper by briefly studying some properties of the resulting conformal nets.

\section{Preliminaries}\label{preliminaries}

Here we shall review the notion of chiral conformal fields with special attention to {\it polynomial energy bounds} and {\it smearing}. In some sense the most appropriate setting for treating these concepts is that of a {\bf unitary vertex operator algebra} (unitary VOA), for which we refer to \cite{CKLW18}.
However, except for Section \ref{w3}, we consider only certain fields and {\it not} the whole unitary VOA. Thus we decided to summarize what we mean by conformal fields {\it without} reference to a unitary VOA. This makes our presentation not only more transparent but also applicable to different settings, including fields in a representation of the VOA. We expect that the techniques developed here might apply to intertwining fields with appropriate modifications.

\subsection{Primary and quasi-primary fields}\label{fields}
Let $V$ be a complex linear space equipped with a representation
$\{C,L_n: n\in \ZZ\}$ of the Virasoro algebra with central charge $C=c\1$
such that $L_0$ is diagonalizable on $V$ with lowest eigenvalue $h\geq 0$ and with all eigenvalues included in the set $\{h,h+1,h+2,\ldots\}$\footnote{This assumption is mainly for simplicity. In general, for applications of e.g.\! Proposition \ref{pr:globaltolocal} --
since the field operators always change the eigenvalues by an integer --
one could decompose the space $V$ into invariant components where this assumption does hold and then use the fact that the estimate of the cited proposition is independent of the value of $h$.}. That is, we have 
\begin{align}\label{eq:virasoro}
{[} L_n,L_m{]} &= (n-m) L_{n+m} +\frac{c}{12}(n^3-n)\delta_{n,-m}\1
\end{align}
for $n,m\in\ZZ$.
We further assume that $V=\bigoplus_{n=0}^\infty V_{h+n}$ (algebraic direct sum) where $V_{n+h}\equiv {\rm Ker}(L_0-(n+h)\1)$.

A family of linear operators $\{\phi_n\}_{n \in \ZZ}$ on $V$ are said to form the ({\bf{Fourier}}) {\bf modes of a field} on $V$ if for every $v\in V$, there is an $n_v$ such that $\phi_n v=0$ whenever $n\ge n_v$.
We shall further say that $\{\phi_n\}_{n\in\ZZ}$ are the modes of a {\bf primary field of conformal dimension $d$} if
\begin{align}
\label{primary_commrel}
[L_n,\phi_m] = ((d-1)n-m) \phi_{n+m}
\end{align}
for all $n,m\in\ZZ$.
If the above relation holds for all $m\in \ZZ$ and $n=-1,0,1$
(but not necessarily for $|n|>1$) then we say that our field is {\bf quasi-primary}. In either case, the 
(abstract) field associated to these modes is the formal series 
$\phi(z)=\sum_{n\in \ZZ}\phi_{(n)}z^{-n-1}$
where $\phi_{(n)} = \phi_{n-d+1}$. 

In what follows, we shall always assume that there is an inner product (a positive-definite sesquilinear form) $\<\cdot, \cdot \>$ making the given representation of the Virasoro algebra unitary, that is, 
\[
\langle \Psi, L_n \Phi\rangle = \langle L_{-n}\Psi, \Phi\rangle
\]
for all $\Phi, \Psi\in V$. We call it a unitary structure on $V$.
Note that the assumption of a unitary structure 
implies that in particular the central charge $c$ must have a positive value.

Given a unitary structure, we say that our primary (or quasi-primary) field $\phi(z)$ is {\bf hermitian} if it holds that $\<\phi_n\Psi, \Phi\> = \<\Psi, \phi_{-n}\Phi\>$
for all $\Psi, \Phi \in V$.
Note that the representation of the Virasoro algebra (given together with $V$) itself defines a hermitian quasi-primary field
of conformal dimension $2$: the formal series $L(z) = \sum_n L_{n-1} z^{-n-1} = \sum_n L_n z^{-n-2}$ is called the {\bf stress-energy} field.

\subsection{Energy bounds and smeared fields}\label{energybounds}

In the setting described in the previous Section, let us consider now a primary (or quasi-primary) field of conformal dimension $d$ with Fourier modes $\{\phi_n\}_{n \in \ZZ}$. We say that
our field satisfies {\bf polynomial energy bounds} if there exist some constants $C,r,s>0$ such that
\begin{equation}
\label{polyEB}
 \|\phi_n \Psi\| \leq C (1+|n|^r)\|(L_0+\1)^s \Psi\|
\end{equation}
for all $\Psi\in V$ and all $n\in \ZZ$, c.f.\! \cite[Section 6]{CKLW18}.
In this case, the restriction of the operator $\phi_n$ to $V_k$ is a bounded map from $V_k$ to $V_{k-n}$
and hence the (formal) adjoint $(\phi_n)^\dagger$ (as a linear operator on $V$) exists and it satisfies
\[
 \langle \Psi, \phi_n \Phi\rangle = \langle (\phi_n)^\dagger\Psi,\Phi\rangle 
\]
for all $\Psi,\Phi\in V$ and one has that $(\phi^\dagger)_n:=(\phi_{-n})^\dagger$ form
the Fourier modes of a (possibly new) primary (or quasi-primary) field of conformal dimension equal to that of $\phi$
(if our $\phi$ is hermitian, then $\phi^\dagger = \phi$).
Note that polynomial energy bounds of $\phi^\dagger$ follow from those of $\phi$:
for $\Psi_m \in V_m$, we have
\begin{align*}
 \|(\phi^\dagger)_n\Psi_m\| &= \sup_{\Psi_{m+n} \in V_{m+n}, \|\Psi_{m+n}\| = 1}\<\Psi_{m+n}, (\phi_{-n})^\dagger\Psi_m\>
 = \sup_{\Psi_{m+n} \in V_{m+n}, \|\Psi_{m+n}\| = 1}\<\phi_{-n}\Psi_{m+n}, \Psi_m\> \\
 &\le C(1+|n|)^r (m+n+1)^s\|\Psi_m\| \le C(1+|n|)^{r+s} (m+1)^s\|\Psi_m\| \\
 &= C(1+|n|)^{r+s} \|(L_0+\1)^s\Psi_m\|
\end{align*}
and the bounds for a general vector $\Psi = \sum_m \Psi_m$ follow because $(\phi^\dagger)_n\Psi_m$
are orthogonal with each other for different $m$.

An important consequence of polynomial energy bounds is that it makes possible
to ``smear'' the field. That is, for any smooth function $f:S^1\to\CC$ with Fourier coefficients
$\hat{f}_n := \frac1{2\pi}\int_{-\pi}^{\pi}f(e^{i\theta})e^{-in\theta}d\theta$,
the sum
\[
 \phi(f) := \sum_{n\in \ZZ} \hat{f}_n \phi_n
\]
is strongly convergent on $V$: $\sum_{n\in \ZZ} \hat{f}_n \phi_n \Psi$ is a vector of
the Hilbert space $\mathcal{H}$ obtained by the norm-completion of $V$. In this way we get a closable operator
(since its formal adjoint, $\sum_{n\in \ZZ} \overline{\hat{f}_n} (\phi_n)^\dagger$, is also well-defined on $V$ and hence densely defined in $\cH$)
whose closure --- for simplicity of notations --- we shall still denote by $\phi(f)$. We have that
\[
\phi(f)^*\supset \phi^\dagger(\overline{f})
\]
showing that in the particular case when our field is hermitian $f$ is real valued, $\phi(f)$ is \textit{symmetric}. Note that even in this case, \textit{self-adjointness} -- i.e. that  $\phi(f)^*= \phi^\dagger(\overline{f})$ -- does not seem to 
follow in general. On the other hand, $\phi(f)$ always preserves
the set of {\bf smooth vectors} $C^\infty(L_0) := \bigcap_{n\in \NN}\dom(L_0^n)$, where by a slight abuse of notation, we denoted $L_0$ and its closure -- whose domain also contains vectors outside $V$ -- by the same symbol. Actually, more than just the invariance of the domain $C^\infty(L_0)$, we have that the map
\begin{align}
\label{eq:distrcont}
C^\infty(S^1, \CC)\ni f \mapsto \phi(f)\Psi \in C^\infty(L_0)
\end{align}
is a strongly continuous linear map for every $\Psi\in C^\infty(L_0)$. 
For details and proofs on these statements (and also on other facts in this Section) see e.g.\! 
\cite[Section 6]{CKLW18}.

The stress-energy automatically satisfies the following {\it linear} energy bound:
\begin{align*}
\|L_n \Psi\|\leq \sqrt{1+\frac{c}{12}}(1+|n|^{\frac{3}{2}})\|(\1+L_0)\Psi\|
\end{align*} 
for every $\Psi\in V$ (see \cite[(3.23)]{WeinerThesis}).
Although -- following VOA notations --  we denoted the stress-energy field by $L(z)$, we shall use $T(f)$
for the corresponding smeared (and closed) operators. 

\subsection{Diffeomorphism covariance}

It is well-known that, as a consequence of the linear energy bound, hermitianity of $L(z)$ implies $T(f)^*=T(\overline{f})$
\cite[Proposition 2]{Nelson72}, as $[L_0, T(f)] = iT(f')$ satisfies a linear energy bound as well;
in particular, when $f$ is a real function, $T(f)$ is self-adjoint.
Moreover, the thus obtained self-adjoint operators always generate a positive energy projective
representation of the group $\diff_+(S^1)$ of orientation-preserving diffeomorphisms. Again, mainly to fix notations and conventions, let us quickly recall what does this exactly mean. 

For a differentiable map $\gamma:S^1\to S^1$ we will consider its derivative to be a real-valued
function; in our convention $\partial \gamma(e^{i\theta}) = -i\frac{d}{d\theta} \log(\gamma(e^{i\theta}))$ and in this way $\partial\gamma$ is a strictly positive function for every $\gamma\in\diff_+(S^1)$. For every smooth function $f:S^1\to\RR$ there exists a unique differentiable one-parameter subgroup $t\mapsto \gamma^f_t \in \diff_+(S^1)$ such that 
\begin{align*}
\frac{d}{dt}\gamma^f_t (z)|_{t=0} = i z f(z)
\end{align*}
for every $z\in S^1\equiv\{w\in \CC:\, |w|=1\}$; this is precisely the one-parameter subgroup usually referred as the {\it exponential} when one thinks of $f$ as a vector field.
With this convention, the rotations are generated by the constant function $1$.

With these notations one has that there is a unique strongly continuous projective unitary representation $U$ of $\diff_+(S^1)$ such that $e^{iT(f)t} = U(\gamma^f_t)$ (in the projective sense, i.e.\! up to a scalar) for every $f:S^1\to\RR$ smooth function and $t\in \RR$, and the following properties hold:
\begin{enumerate}[{(}a{)}]
\item\label{li:fieldcovariance} $U(\gamma)T(f)U(\gamma)^* = T((\partial \gamma\circ \gamma^{-1})(f\circ \gamma^{-1})) + r \1$
with the constant $r$ depending on the central charge $c$, the diffeomorphism $\gamma$ 
and the function $f$ \cite[Proposition 3.1]{FH05}\footnote{The constant $r$ is given in an explicit manner in \cite[(3.27)]{FH05}:
$r = -\frac c{24}\int_{S^1}\{\gamma, z\}f(z)dz$, where $\{\gamma, z\}$ is the Schwarz derivative.}; 
\item for every $f\geq 0$, the self-adjoint operator $T(f)$ is bounded from below \cite[Theorem 4.1]{FH05};

\item\label{li:continuityT} $C^\infty(S^1,\RR)\ni f\mapsto T(f)$ is continuous in the strong resolvent sense,
and in particular, $C^\infty(S^1,\RR^+_0)\ni f\mapsto {\rm min}({\rm Sp}(T(f)))\in \RR$ 
is continuous \cite[Proposition 4.5]{CW05}.
\item\label{li:boundT} for $f_1,\cdots ,f_n \in C^\infty(S^1,\CC)$, there is $t>0$
such that\footnote{This can be
proved from the linear enegy bound and induction: The $n=1$ case is exactly the linear energy bound. Then,
$\|T(f_1)\cdots T(f_n) \Psi\| \le t\|(L_0+\1)T(f_2)\cdots T(f_n) \Psi\|$
and $(L_0+\1)T(f_2)\cdots T(f_n) \Psi = \sum_{j=2}^n T(f_2)\cdots T(f_j')\cdots T(f_n)\Psi +
T(f_2)\cdots T(f_n) (L_0+\1)\Psi$, and we can apply the assumption of induction to the last expression, as $\Psi \in C^\infty(L_0)$.
For the second claim, we take $n > s$, then we have $(T(f)+q\1)^{2n} \le t(L_0+\1)^{2n}$ for some $t>0$,
and as the function $f(x) = x^{\frac sn}$ is an operator monotone \cite[Theorem 4.1]{Simon19Loewners},
we have $(T(f)+q\1)^{2s} \le (t(L_0+\1))^{2s}$ by \cite[Theorem 2.9]{Simon19Loewners}.}
$\|T(f_1)\cdots T(f_n) \Psi\| \le t\|(\1+ L_0)^n \Psi \|$ for every $\Psi\in C^\infty(L_0)$.
Furthermore, if $f \in C^\infty(S^1, \CC)$ and $q \in \RR$ is such that $T(f) + q\1$ is positive and if $s > 0$,
then there is $t$ such that $\|(T(f)+q\1)^s\Psi\| \le t\|(L_0+\1)^s\Psi\|$ for all .
\end{enumerate}  

Property (\ref{li:fieldcovariance}) essentially says that, up to a multiplicative constant, $T$ transforms under diffeomorphisms like vector fields do.

If $\phi(z)$ is a primary field of conformal dimension $d$ satisfying polynomial energy bounds, the commutator formula (\ref{primary_commrel}) results in the following:
\begin{align}\label{eq:covariance}
U(\gamma) \phi(f) U(\gamma)^* = \phi((\partial \gamma\circ \gamma^{-1})^{d-1}(f\circ \gamma^{-1}))
\end{align}
for every $f\in C^{\infty}(S^1,\CC)$ and $\gamma\in\diff_+(S^1)$; see e.g.\!\! \cite[Proposition 6.4]{CKLW18} for details. If $\phi$ is only quasi-primary, then the above transformation rule does not necessary hold for all $\gamma\in \diff_+(S^1)$ but it is still true for elements of the subgroup
$\mob \subset \diff_+(S^1)$ consisting of transformations of the form
\begin{align*}
S^1 = \RR\cup\{\infty\}\ni z \mapsto \frac{az+b}{cz+d} ,\;\;\;\; \left(\begin{matrix}
a & b \\ c & d
\end{matrix}\right)\in {\rm SL}(2,\RR).
\end{align*} 
Note that the $2d-1$ dimensional subspace of $C^\infty(S^1,\CC)$ consisting of linear combinations 
of the $2d-1$ functions $z\mapsto z^k$ $(k\in\ZZ, |k|<d)$ is invariant under the action
\begin{align*}
\mob\ni \gamma:  f\mapsto (\partial\gamma\circ\gamma^{-1})^{d-1}(f\circ\gamma^{-1}),
\end{align*}
and that this is actually a $2d-1$ dimensional irreducible representation of $\mob$. In particular, if the conformal dimension of our quasi-primary field is $d$, then
for any $k\in\ZZ, |k|<d$
\begin{align*}
{\rm Span}\{ U(\gamma)\overline{\phi_k} U(\gamma)^*|_{V} :\; \gamma\in \mob \}
= {\rm Span}\{\phi_j: \, j\in \ZZ, |j|< d\},
\end{align*}
where $\overline{\phi_k}$ is the closure of $\phi_k$
(we need the closure, because $U(\gamma)^*$ does not preserve $V$).

\subsection{Strong locality}
The von Neumann algebra generated by a collection of bounded operators is the smallest von Neumann algebra containing the collection.
Here we need to consider the von Neumann algebras generated by some (possibly) unbounded operators in an appropriate sense.

A densely defined closed operator $X$ on a Hilbert space $\mathcal H$ has a unique {\it polar decomposition}: that is, there exists a unique partial isometry $U$ with ${\rm Ker}(U)= {\rm Ker}(X)$ such that $X=UA$, where $A$ is the positive self-adjoint operator $A = \sqrt{X^*X}$,
which further admits the spectral decomposition; 
the thus obtained projections and the partial isometry $U$ can be viewed as the bounded building blocks of $X$.
One says that $X$ is \textbf{affiliated} to a von Neumann algebra $\mathcal M$ when $\M$ contains $U$ and all of
the spectral projections of $A$. Of course, when $\|X\|<\infty$, $X$ is affiliated to $\M$ if
and only if $X\in \M$. We define the von Neumann algebra $W^*(\{T_\alpha\})$ generated by a collection $\{T_\alpha\}$
of densely defined closed operators as the smallest von Neumann algebra to which all $T_\alpha$'s are affiliated.  

Let now $V$ be a complex inner product space carrying a unitary representation of the Virasoro algebra 
 $\{L_n: n\in\mathbb Z\}$ of the type explained in Section \ref{fields}, $\mathcal H =\overline{V}$ the Hilbert space completion of $V$ and $\phi$ a polynomially energy bounded field on $V$. For each nonempty, nondense, open connected interval of the circle $I\subset S^1$, consider the von Neumann algebra
\[
\A_\phi(I)\equiv W^*(\{\phi(f) | f\in C^\infty(S^1,\CC), {\rm Supp}(f)\subset I \}).
\]
We say that two fields $\phi,\tilde{\phi}$ are \textbf{relatively local} (to each other)
if there is $N \in \NN$ such that $(z-w)^N[\phi(z), \tilde \phi(w)] =0$ as formal series.
By \cite[Appendix A]{CKLW18}, for fields satisfying polynomial energy bounds,
this is equivalent to the vanishing of the commutator $[\phi(f),\tilde{\phi}(\tilde{f})]$ on the
set of smooth vectors $C^\infty(L_0)$ for every pair of smooth functions $f,\tilde{f}\in C^\infty(S^1,\CC)$
with disjoint support; i.e.\! the ``weak'' (or algebraic) relative locality of the two fields.
Furthermore, two fields $\phi,\tilde{\phi}$ are said to be \textbf{strongly relatively local} (to each other)
if $\A_\phi(I)$ and $\A_{\tilde{\phi}}(\tilde{I})$ are commuting von Neumann algebras whenever $I\cap \tilde{I}=\emptyset$.
We say that a family of fields $\{\phi^{(j)}\}$ is \textbf{local} (respectively \textbf{strongly local}) if any pair of these fields (including pairs of the same field) is relatively local (respectively strongly relatively local). If a family of fields is strongly local then it is also local, see e.g.\! \cite[Proposition 2.1]{CKLW18} but the converse is not known to hold in general and our purpose is to give new sufficient conditions on local families in order to get strong locality. 

By an adaptation of \cite[Lemma 6.5]{CKLW18} (we do not need vertex operator $Y$, but only
the rotation covariance of $\phi$), one has that a bounded operator $B'$ is in the commutant of 
of $\A_\phi(I)$ if and only if 
$$
\langle B'\Psi,\phi(f)\Phi \rangle =
\langle \phi^\dagger(\overline{f})\Psi,B'^*\Phi \rangle
$$
for all $\Psi,\Phi\in V$ and $f\in C^\infty(S^1,\CC)$ with support in $I$. Taking also into account that $B'$ is an element of the commutant of a certain von Neumann algebra if and only if $B'^*$ is so, the above characterization implies that
$$
\A_\phi(I) = \A_{\phi^\dagger}(I) =
\A_{\phi+\phi^\dagger}(I) \bigvee \A_{i(\phi -\phi^\dagger)}(I).
$$
This means that when checking strong relative locality, it is enough to work with hermitian fields:
$\phi$ and $\tilde{\phi}$ are strongly relatively local if and only if any of the hermitian fields 
$\phi+\phi^\dagger$ and $i(\phi-\phi^\dagger)$ is
strongly relatively local to any of hermitian fields
$\tilde{\phi}+\tilde{\phi}^\dagger$ and $i(\tilde{\phi}-\tilde{\phi}^\dagger)$.

    \section{Local energy bounds and strong locality}
    \label{sec:main}
    Let $V$ and fields be as in Section \ref{fields} and $\{L_n\}$ be the associated representation of the Virasoro algebra
    (but we do not necessarily assume the existence of vacuum, unless otherwise specified).

    \subsection{From global to local energy bounds}
    \label{sec:globaltolocal}
    Here we show that an energy bound regarding a single Fourier 
    component of a primary field of the \textit{optimal degree} can be turned into 
    a local energy bound regarding smeared fields.
    
    \begin{proposition}\label{pr:component}
    Let $\phi$ be a hermitian primary field of conformal dimension $d>1$ and $s > 0$.
    The following properties are equivalent:
    \begin{enumerate}[{(}1{)}]
    \item\label{li:boundzero} there exist $k\in \ZZ$ and a constant $C>0$ such that
    \begin{align}\label{eq:bound0}
    \|\phi_k \Psi\| \leq C \|(L_0+\1)^s \Psi\|
    \end{align}
    for every $\Psi\in V$,
    \item\label{li:polynomial} $\phi$ satisfies polynomial energy bounds and 
    $\phi(f)(L_0+\1)^{-s}$ is a bounded operator for every $f\in C^\infty(S^1, \CC)$.
    \end{enumerate}
    \end{proposition}
    \begin{proof}
    
    (\ref{li:polynomial}) implies (\ref{li:boundzero}) since $\phi(\mathfrak e_k) = \phi_k$
    (we denote the closure with the same symbol), where $\mathfrak e_k(z) = z^k$.
  
    Let us assume (\ref{li:boundzero}).
    Then $\phi$ satisfies a certain polynomial energy bound: Indeed,
    if $k\neq 0$ we have
    $\phi_0 = -\frac1{kd}[L_{-k}, \phi_k]$ (on $V$) and we obtain a bound on $\phi_0$.
    If $k=0$, then \eqref{eq:bound0} is directly a bound on $\phi_0$.
    We have
    $\phi_n = \frac{1}{(d-1)n}[L_n,\phi_0]$ for $n\in \ZZ$
    and hence the field $\phi$ satisfies polynomial energy bounds (because $L_n$ satisfies a linear energy bound
    \cite[below Lemma 4.1]{CW05})
    and one can consider the smeared field $\phi(f)$ for every $f \in C^\infty(S^1, C)$.

    We prove the claimed energy bound in four steps:
    first for a certain $f > 0$, next for the constant function $1$, and
    arbitrary $f > 0$ and finally for $f \in C^\infty(S^1, \CC)$.
    
    If $k=0$, the first and second steps are done. Let us consider other cases.
    (\ref{li:boundzero}) implies
    $\phi_{-k}(L_0+(k+1)\1)^{-s}$ is bounded.
    Indeed, if $k>0$,
    \[
     \phi_{-k}(L_0+(k+1)\1)^{-s} = ((L_0+(k+1)\1)^{-s}\phi_k)^\dagger = (\phi_k(L_0+\1)^{-s})^\dagger
    \]
    and the right-hand side is bounded by \eqref{eq:bound0}.
    Therefore, we have
    $\|\phi_{-k} \Psi\| \leq C\|(L_0+\1)^s\Psi\|$ with a possibly different $C$
    because $(L_0+(k+1)\1)^{s}(L_0+\1)^{-s}$ is bounded.
    If $k<0$,
    \[
     \phi_{-k}(L_0+\1)^{-s} = ((L_0+\1)^{-s}\phi_k)^\dagger = (\phi_k(L_0+(1-k)\1)^{-s})^\dagger
    \]
    and the right-hand side is bounded by \eqref{eq:bound0} and a similar argument as above.
    
    Furthermore,
    $\|\phi(\mathfrak c_k) \Psi\| \leq C \|(L_0+\1)^s \Psi\|$ with again a possibly different $C$,
    where $\mathfrak c_k(e^{i\theta}) = \cos k\theta$.
    As $\phi$ is primary, we have
    $U(\gamma) \phi(f) U(\gamma)^* = \phi((\partial \gamma\circ \gamma^{-1})^{d-1}(f\circ \gamma^{-1}))$
    and we can choose $\g$ such that $(\partial \gamma\circ \gamma^{-1})^{d-1}(\mathfrak c_k\circ \gamma^{-1})$
    has small negative part (by choosing $\g$ which ``shrinks'' the part where $\mathfrak c_k$ is negative
    and hence $\partial \g$ is small there).
    Then by further applying rotations, we can find a finite set $\g_1,\cdots,\g_\ell$ such that
    $\sum_{j=1}^\ell U(\g_j) \phi(\mathfrak c_k) U(\g_j)^* = \phi(f)$ and this $f$ is \emph{strictly positive}.
    For this $f$,
    \begin{align*}
    \|\phi(f)\Psi\| &= \left \|\sum_{j=1}^\ell U(\g_j) \phi(\mathfrak c_k) U(\g_j)^*\Psi \right\| \le \sum_{j=1}^\ell \|\phi(\mathfrak c_k) U(\gamma_j)^*\Psi\| \\
    &\le \sum_{j=1}^\ell C \|(L_0+\1)^s U(\g_j)^*\Psi\| = \sum_{j=1}^\ell C\|(U(\g_j)(L(1) + \1) U(\g_j)^*)^s \Psi\| \\
    &= \sum_{j=1}^\ell C \|(T(\partial \g_j\circ\g_j^{-1}) +q_j \1)^s \Psi\|,
    \end{align*}
    where $q_j \in \RR$,
    and the last expression can be bounded by $C_f \|(L_0 + \1)^s \Psi\|$
    with some $C_f > 0$ by property \eqref{li:boundT}.
    
    With $f$ in the previous paragraph, let us write $f=g^{d-1}$ where $g$ is again a \emph{strictly positive} smooth function on the circle.
    Then the solution of the ordinary differential equations
    \[
     \sigma'(s) = N_g g(e^{i\sigma(s)}), \qquad N_g = \frac1{2\pi}\int_{0}^{2\pi} \frac{1}{g(e^{i\theta})} d\theta,
    \]
    with the initial condition $\sigma(0)=0$
    can be given by $\sigma(s) = \tau^{-1}(s)$ (the inverse function of $\tau$), where
    $\tau(t) = \int_0^t \frac1{N_g g(e^{i\theta'})}d\theta'$, which is monotonically increasing with derivative $\tau'(\theta) = \frac1{N_g g(e^{i\theta})} > 0$
    and $\tau(t+2\pi) = \tau(t) + 2\pi$. Therefore $\sigma(\theta+ 2\pi) = \sigma(\theta) + 2\pi$ as well,
    and we can define a diffeomorphism by $\gamma(e^{i\theta}) = e^{i\sigma(\theta)}$
    (the constant $N_g$ is needed here to assure that $\gamma(e^{i(\theta + 2\pi)}) = e^{i\sigma(\theta+2\pi)} = e^{i\sigma(\theta)} = \gamma(e^{i\theta})$)
    satisfying
    \[
     \partial \gamma \circ \gamma^{-1}(e^{i\theta}) = N_g g(e^{i\theta}),\;\; \gamma(1)=1.
    \]
    Then
    \[
     U(\gamma)\phi_0 U(\gamma)^* = 
     U(\gamma)\phi(1) U(\gamma)^* =\phi((N_g g)^{d-1}) = N^{d-1}_g\phi(f) 
    \]
    and in a similar manner $U(\gamma)^*L_0 U(\gamma) = 
    U(\gamma)^*T(1) U(\gamma) = T(\tilde f) + q_\g\1$ for some $\tilde f > 0, q_\g \in \RR$, showing that
    \begin{align*}
    \|\phi_0\Psi\| &= N_g^{d-1}\|U(\gamma)^* \phi(f) U(\gamma)\Psi\|
     = N_g^{d-1}\|\phi(f) U(\gamma)\Psi\| \\
    &\leq C_f N_g^{d-1} \|(L_0 + \1)^s U(\gamma)\Psi\| \\
    &\le \tilde C \|(L_0 + (q_\g + 1)\1)^s \Psi\|
    \end{align*}
    and by property (\ref{li:boundT}) with a further possibly different $C$,
    we have $\|\phi_0\Psi\| \le C \|(L_0 + \1)^s \Psi\|$.
    
    Next, let $f$ an arbitrary smooth strictly positive function.
    We set $g = f^{\frac1{d-1}}$ and take $\gamma$ in the same way as above. Then we have 
    $U(\gamma)\phi_0 U(\gamma)^* = N^{d-1}_g\phi(f)$, $U(\gamma)L_0 U(\gamma)^* = T(\tilde f) + q_\g$
    for some $\tilde f, q_\g$ (different from those in the previous paragraph) and
    \begin{align*}
    \|\phi(f)\Psi\| &= N_g^{-(d-1)}\|U(\gamma)\phi_0 U(\gamma)^*\Psi\| \\
    &\leq C N_g^{-(d-1)} \|(L_0 + \1)^s U(\gamma)^*\Psi\|
    = C N_g^{d-1} \|(T(\tilde f) + (q_\g + 1)\1)^s \Psi\| \\
    &\le C_f \|(L_0 + \1)^s \Psi\|
    \end{align*}
    for some $C_f > 0$, again by property (\ref{li:boundT}).

    In other words, $\|\phi(f)(L_0 + \1)^{-s}\|<\infty$ whenever $f$ is strictly positive.
    This concludes the proof, because
    every $f:S^1\to \CC$ is a linear combination of four strictly positive smooth functions.
    \end{proof}
    
    We shall now see that a nontrivial (i.e.\! non constant zero) primary field of conformal dimension
    $d>1$ can satisfy a bound of the form
    $\|\phi_0 \Psi\| \leq C \|(L_0 + \1)^s \Psi\|$
    \emph{at best} with degree $s \ge d-1$.
    Moreover, in case the degree is \textit{precisely} $d-1$ (hence it is \textbf{optimal}),
    such a global energy bound (formulated in terms of $L_0$) can be promoted to \textit{local energy bound},
    where the field smeared on a local test function is bounded by some polynomial of the stress-energy tensor
    smeared on a (in general different) local test function.
    \begin{proposition}\label{pr:globaltolocal}
    Let $\phi$ be a hermitian primary field of conformal dimension $d>1$ and $s,C$ positive constants such that
    $\|\phi_0 \Psi\| \leq C \|(L_0 + \1)^s \Psi\|$ for every $\Psi\in V$. If $s<d-1$ then $\phi=0$, and if 
    $s=d-1$, then for every $g\geq 0$ smooth function
    and every $\Psi\in C^\infty(L_0)$,
    \begin{align}\label{eq:leb}
    \|\phi(g^{d-1}) \Psi \| \leq \,C\,  \|(T(g)+r_g \1)^{d-1}\Psi\|
    \end{align}
    where the constant $r_g = -{\rm min}({\rm Sp}(T(g)))$ is universal in the sense that
    it depends only on $g$ and the central charge $c$ (but not on the lowest energy level $h$).
    \end{proposition}
    \begin{proof}
    By Proposition \ref{pr:component}, $\phi$ satisfies polynomial energy bounds.
    Furthermore, if $\{g_n\}$ are strictly positive smooth functions,
    then there are $\g_n \in \diff_+(S^1)$ such that $\g_n'\circ \g_n = N_{g_n}g_n, \g_n(1)=1$
    and hence $U(\g_n)\phi_0 U(\g_n)^* = N_{g_n}^{d-1}\phi(g_n^{d-1})$.
    By property \eqref{li:fieldcovariance}, it also holds that
    $U(\g_n)L_0 U(\g_n)^* = T(N_{g_n}g_n) + q_n\1$,
    where
    $q_n = h - {\rm min}({\rm Sp}(T(N_{g_n} g_n))) =: h + N_{g_n}r_{g_n}$.
    
    By a straightforwards estimate (as in the proof of Proposition \ref{pr:component}),
    we have
    \begin{align}\label{eq:quasilocal}
    \|\phi(g_n^{d-1})\Psi\| \leq C N_{g_n}^{s-(d-1)} \left\|\left(T(g_n)+\left(r_{g_n}+\frac{h+1}{N_{g_n}}\right)\1\right)^s \Psi\right\|
    \end{align}
    for every $\Psi\in C^\infty(L_0)$.
    
    For a positive function $g$ supported in a proper interval, we can 
    take a sequence of strictly positive functions $g_n$ $(n\in \NN)$ converging to $g$
    in $C^\infty(S^1)$ by adding constants converging to $0$.
    Since our $g$ is ``local'' (its support is not the whole circle), $N_{g_n} \to \infty$
    and by the properties (\ref{li:boundT}) and (\ref{li:continuityT})
    $\|(T(g_n)+(r_{g_n}+\frac{h+1}{N_g})\1)^s \Psi\|\to \|(T(g)+r_{g}\1)^s \Psi\|$ and 
    $\|\phi(g_n^{d-1})\Psi\|  \to \|\phi(g^{d-1})\Psi\|$ by (\ref{eq:distrcont}).
    
    If $s<d-1$, then by comparing the sides of \eqref{eq:quasilocal}, $\|\phi(g^{d-1})\Psi\|=0$.
    However, $\Psi\in C^\infty(L_0)$ was arbitrary and 
    local, nonnegative, smooth functions linearly span the set $C^\infty(S^1, \CC)$ so in this case
    our field is trivial: $\phi=0$.
    
    On the other hand, if $s=d-1$ then the limit results in the desired inequality.
    The only thing to comment about is the fact that by \cite[Theorem 4.1]{FH05},
    when $g$ is local, the minimum of the spectrum of $T(g)$ is independent of $h$. 
    \end{proof}
    
    We say that a primary field $\phi$ satisfies a \textbf{global energy bound} if
    there is $C, s > 0$ such that $\|\phi_0\Psi\| \le C \|(L_0 + \1)^s \Psi\|$.
    A global energy bound is said to be \textbf{optimal} if $s = d-1$.
    We say $\phi$ satisfies a \textbf{local energy bound} if
    for any $f \ge 0$ smooth function, there is $C, s, q > 0$ such that
    \begin{align}\label{eq:lebg}
     \|\phi(f^{d-1})\Psi\| \le C \|(T(f) + q)^s \Psi\|.
    \end{align}
    Proposition \ref{pr:globaltolocal} can be summarized as follows:
    if $\phi$ satisfies an optimal global energy bound, then it satisfies a local energy bound
    with $s=d-1$.

    \subsection{Improving global energy bounds}\label{improving}
    
    \textit{A priori}, it is not always easy to have the optimal global energy bounds $s=d-1$ for a given primary field $\phi$.
    The next Proposition gives a way to obtain a global energy bound from a bound on a commutator.
    
    \begin{proposition}\label{pr:improving}
     Let $\phi$ be a hermitian (quasi-)primary field and assume that
     \[
      \|[\phi_{-k}, \phi_k](L_0 + \1)^{-\b}\| < C
     \]
     for some $k\neq 0, \b > 0, C > 0$. Then, with $C'' =  (1+|k|)^{\frac{\beta+1}2}\sqrt{\frac{C}{|k|(\b + 1)}} > 0$,
     $\|\phi_k(L_0 + \1)^{-\frac{\b+1}2}\| < C''$.
    \end{proposition}\label{pr:canonicalbound}
    \begin{proof}
     We first assume that $k > 0$.
     Put $\g(n) = \|\phi_k|_{\cH_n}\| < \infty$.
     In this proof, $\Psi_\ell$ denotes a generic vector in $V_\ell$.
     We have
     \begin{align*}
      \|\phi_k\Psi_{n+k}\|^2 &= \sup_{\|\Psi_n\| = 1}|\<\Psi_n, \phi_k\Psi_{n+k}\>|^2
      = \sup_{\|\Psi_n\| = 1} |\<\phi_{-k}\Psi_n, \Psi_{n+k}\>|^2 \\
      \|\phi_{-k}\Psi_n\|^2 &= \<\Psi_n, [\phi_k, \phi_{-k}]\Psi_n\> + \|\phi_k\Psi_n\|^2
     \end{align*}
     Therefore, by assumption it holds that
     \[
      \g(n+k)^2 \le \g(n)^2 + C(n+1)^\b.
     \]
     As we have $k>0$, it holds that $\g(0) = \g(1) = \cdots = \g(k-1) = 0$,
     it is easy to obtain the estimate
     \[
      \g(n)^2 \le \frac{C}{k(\b + 1)} (n+1)^{\b + 1},
     \]
     from which the desired estimate follows with $C' = \sqrt{\frac{C}{k(\b + 1)}}$.

     Next we consider the negative case.
     Note that the assumption $\|[\phi_{-k}, \phi_k](L_0 + \1)^{-\b}\| < C$
     remains invariant as $k$ is replaced by $-k$.
     Therefore, with $k>0$, we will show the desired inequality with $k$ replaced by $-k$.
     
     As for $\phi_{-k}$,
     \begin{align*}
       \|\phi_{-k}(L_0 + \1)^{-\frac{\b+1}2}\Psi_n\|^2
       &=\|\phi_{-k}(n + \1)^{-\frac{\b+1}2}\Psi_n\|^2 \\
       &\le (1 + k)^{\b+1}\|\phi_{-k}(n + \1 + k)^{-\frac{\b+1}2}\Psi_n\|^2 \\
       &= (1 + k)^{\b+1}\|(L_0 + \1)^{-\frac{\b+1}2}\phi_{-k}\Psi_n\|^2 \\
       &\le (1 + k)^{\b+1}\|(L_0 + \1)^{-\frac{\b+1}2}\phi_{-k}\|^2 \cdot\|\Psi_n\|^2 \\
       &\le (1 + k)^{\b+1}\|\phi_{k}(L_0 + \1)^{-\frac{\b+1}2}\|^2 \cdot\|\Psi_n\|^2 \\
       &\le (1 + k)^{\b+1}(C')^2 \|\Psi_n\|^2.
     \end{align*}
     As $\phi_{-k}(L_0 + \1)^{-\frac{\b+1}2}\Psi_n$ are mutually orthogonal for different $n$,
     for a general $\Psi$ we obtain $\|\phi_{-k}(L_0 + \1)^{-\frac{\b+1}2}\Psi\|^2 \le (1 + k)^{\b+1}(C')^2 \|\Psi\|^2$,
     which is what we had to prove
     with $C'' =  (1+k)^{\frac{\beta+1}2}\sqrt{\frac{C}{k(\b + 1)}}$.
    \end{proof}
    
    \subsection{Strong locality through local energy bounds}
    \label{sec:strong}
    
    Let $\phi$ be a primary field of dimension $d$. Then it holds that
    \begin{align}\label{eq:primary_smeared}
     [\phi(f), T(g)] = \phi((d-1)fg'-f'g).
    \end{align}
    In particular, if we take $g \ge 0, f = g^{d-1}$, then $[\phi(f), T(g)] = 0$.
    
    There are results which allow us to conclude strong commutativity from weak commutativity under certain circumstances.
    The first of which is due to Glimm and Jaffe (see \cite[Theorem 19.4.4]{GJ87}).
    For our purpose, we utilize the following special case of Driessler-Fr\"ohlich theorem \cite{DF77}
    (more precisely a variation of it \cite[Theorem C.2]{Tanimoto16-1} with vanishing commutators).
    \begin{theorem}\label{th:DF}
     Let $H\ge \1$ be a positive self-adjoint operator, $A$ and $B$ symmetric operators on a core $\mathscr{D}$ of $H$. Assume that
     there is positive number $C$ such that 
     \begin{enumerate}[(1)]
      \item\label{df1} $\|A\Psi\| \le C\|H\Psi\|$ and $\|B\Psi\| \le C\|H\Psi\|$ for all $\Psi \in \mathscr{D}$.
      \item\label{df2} For all $\Psi, \Phi \in \mathscr{D}$,
 \[
  \<H\Psi,A\Phi\> - \<A\Psi, H\Phi\> =  \<H\Psi,B\Phi\> - \<B\Psi, H\Phi\> = 0
 \]
      \item\label{df4} $\<A\Psi, B\Phi\> = \<B\Psi, A\Phi\>$ for all
      $\Psi, \Phi \in \mathscr{D}$.
     \end{enumerate}
     Then $A, B$ are essentially self-adjoint on any core of $H$ and they strongly commute.
    \end{theorem}

    Note that the condition (\ref{df2}) is requiring that the commutators $[A,H],[B,H]$ vanish in the weak sense.
    In the original theorems \cite[Theorem 19.4.4]{GJ87}\cite{DF77}, the most important assumption is that the commutators
    $[A,H], [B,H]$ are bounded by $H$ in a certain way.
    If we have a local energy bound \eqref{eq:lebg}, we can skip this condition
    by taking $H$ that \emph{commutes with $A,B$}
    and conclude strong commutativity of fields from weak commutativity.
    \begin{theorem}\label{th:strong}
    Let $\phi^{(1)}, \phi^{(2)}$ be relatively local hermitian primary fields with conformal dimensions $d_1, d_2 > 1$ and satisfy the local energy bound \eqref{eq:lebg}
    with $s_j = d_j - 1$.
    If $g_1$ and $g_2$ are nonnegative test functions with disjoint supports, then
    $\phi^{(1)}(g_1^{d_1-1})$ and $\phi^{(2)}(g_2^{d_2-1})$ are essentially self-adjoint on $C^\infty(L_0)$ and strongly commute.
    \end{theorem}
    \begin{proof}
    Let $g_1, g_2$ as above.
    Note that then $T(g_1 + g_2)$ is bounded below and it is essentially self-adjoint on $C^\infty(L_0)$.
    Furthermore, on $C^\infty(L_0)$, $\phi^{(1)}(g_1^{d-1})$ and $\phi^{(2)}(g_2^{d-1})$ weakly commute.
    As we have $[\phi^{(j)}(g_j^{d_j-1}), T(g_j)] = 0$ by \eqref{eq:primary_smeared}
    (in the sense of sesquilinear form) and since $g_1$ and $g_2$ have disjoint support,
    it holds that $[\phi^{(j)}(g_j^{d_j-1}), (T(g_1 + g_2) + q)^{d-1}] = 0$,
    where $d = \max\{d_1, d_2\}$ and $q > 0$ such that $T(g_1 + g_2) + q$ is positive
    (such a $q$ exists by \cite[Theorem 4.1]{FH05}).
    
    As $g_1$ and $g_2$ have disjoint supports, $\phi^{(1)}(g_1^{d_1-1})$ and $\phi^{(2)}(g_2^{d_2-1})$
    weakly commute.
    Furthermore, by taking $q$ even larger,
    we may assume that
    $\|T(g_j)^{d-1}\Psi\| \le \|(T(g_1+g_2)+q)^{d-1}\Psi\|$,
    because $T(g_1) + \frac q 2$ and $T(g_2) + \frac q 2$ are positive and strongly commute.
    Now from \eqref{eq:leb} and Theorem \ref{th:DF},
    with $A = \phi^{(1)}(g_1^{d_1-1}), B= \phi^{(2)}(g_2^{d_2-1})$ and $H = (T(g_1+g_2)+q)^{d-1} + \1$,
    we conclude that $\phi^{(1)}(g_1^{d_1-1})$ and $\phi^{(2)}(g_2^{d_2-1})$
    are essentially self-adjoint on $C^\infty(L_0)$ (which is a core of $H$) and strongly commute.
    \end{proof}
    
 \paragraph{Unitary VOA and conformal net.}
 Let us assume that $h=0$ and $V$ contains a distinguished vector $\Omega \in V_0$ and $V$ is generated a set of quasi-primary fields $\{\phi^{(j)}\}$.
 If they satisfy certain properties, $V$ admits a structure of a \textit{simple unitary vertex operator algebra (VOA)},
 generated by the field(s) $\{\phi^{(j)}\}$ and in particular one can define a local family of fields $\{Y(a,z): a \in V\}$ which is maximal in an appropriate sense (the family of \textbf{vertex operators}); we do not recall the definition here, but refer to \cite[Section 4.2]{CKLW18}.
 A VOA $V$ contains a distinguished stress-energy field $L(z)$.
 If there are other hermitian quasi-primary fields satisfying the Virasoro relations \eqref{eq:virasoro}, we call them
  \textbf{Virasoro fields}. A hermitian primary field $Y(a,z) = J(z)$, where $a$ belongs to $V_1$ of the subspace of eigenvectors of $L_0$ with eigenvalue $1$,
 is called a \textbf{current}. Any current satisfies $[J_m, J_n] = m\delta_{m+n}$ \cite[Proposition 6.3]{CKLW18}.
 
 Given a simple unitary VOA, one may try to construct a conformal net (see \cite[Section 3.3]{CKLW18} for the definition) 
 on the Hilbert space completion $\mathcal{H}_V$ of $V$ by
 \[
  \A(I) = W^*(\{Y(a,f):a\in V, f \in C^\infty(S^1, \CC), \supp f\subset I\}),
 \]
 together with the (projective) representation $U$ of $\diff_+(S^1)$ and $\Omega$, but the locality of $\A$ does not follow in general.
 Yet, \cite[Theorem 8.1]{CKLW18}, if the fields $\{\phi^{(j)}\}$ are strongly local, then $\A$ is local and we obtain a conformal net.
 In this case, $V$ is said to be strongly local.
 
    \begin{lemma}\label{lm:localfield}
     Assume that a simple unitary VOA $V$ contains relatively local hermitian primary fields $\phi^{(j)}, j=1,2$  with conformal dimension $d_j>1$
     satisfying either an optimal global energy bound or a local energy bound with $s_j = d_j-1$.
     Then $\phi^{(1)}(f_1), \phi^{(2)}(f_2)$ strongly commute
     if $f_1,f_2$ have disjoint supports. Furthermore, if $a \in V$ is such that $Y(a,z)$ satisfies polynomial energy bounds,  
     $\underline L$ is a Virasoro field  and $J$ is a  current then $Y(a,g)$ strongly commute with $\underline L(f)$  and  
     $J(f)$ if $g$ and $f$ have disjoint supports.  In particular $\phi^{(j)}(f_j)$ strongly commutes with $\underline L(f)$ and  $J(f)$ if $f_j$, $j=1,2$ and $f$ have disjoint supports.     
     
    \end{lemma}
    \begin{proof}
     Let us first note that these fields smeared with positive functions of certain form commute strongly:
     We know from Proposition \ref{pr:globaltolocal} and
     Theorem \ref{th:strong} that $\phi(f_1)$ and $\phi(f_2)$ strongly commute if the test functions $f_1$ and $f_2$ 
     are  of the form $g_1^{d_1-1}$ and $g_2^{d_2-1}$, where $g_1$ and $g_2$ are smooth nonnegative functions 
     \footnote{Not all nonnegative smooth functions have
     smooth $(d-1)$-th root.}. 
      Primary fields satisfy the covariance \eqref{eq:covariance}
     which preserve the set of nonnegative functions of the form $f_j^{d_j-1}$.
     Therefore, we can construct a (local) conformal net $(\A^+, U, \Omega)$ generated by the fields $\phi_1$, and $\phi_2$ smeared with these nonnegative test functions only. This net might not satisfy the cyclicity of vacuum, but this is not necessary.

     By the argument of \cite[Lemma 6.5]{CKLW18} and the proof of \cite[Proposition 6.6]{CKLW18}, for a linear combination 
     $h = \sum_{k=1}^m \alpha_k g_k^{d_j-1}$ with $g_k$ supported in $I$ and $\alpha_k \in \mathbb{C}$,  $k=1,\dots,m$,      
     $\phi^{(j)}(h)$ is affiliated to $\A^+(I)$: For any open interval $\tilde I$ containing the closure $\overline{I}$ of $I$ and 
     $x \in \A^+(\tilde I)'$, we can take
     a mollified operator $x_n$ which preserves the domain $C^\infty(L_0)$, $x_n\to x$ in the strong operator topology
     and commutes with $\phi^{(j)}(g_k^{d_j-1})$,  $k=1,\dots,m$,      
     hence with $\phi^{(j)}(h)$. Therefore, $\phi^{(j)}(h)$ commutes (strongly) with
     any operator in $\A^+(\tilde I)'$, but $\tilde I \supset \overline{I}$ was arbitrary and hence 
     $\phi^{(j)}(h)$ is affiliated with $\bigcap_{\tilde{I} \supset \overline{I}} \A^+(\tilde{I})$ which is equal to $\A^+(I)$ by conformal covariance.

     Now, any smooth function $f$ with support in $I$ can be obtained by a linear combination of functions of the form   
     $g^{d_j-1}$ with $\supp g \subset I$: It is enough to consider the real part, and the real part is a linear combination of
     two smooth nonnegative functions. Any smooth nonnegative function can be written as a difference
     of two function of the form $f^{d_j-1}$, because we can add a smooth positive function with a larger support which has a smooth $(d_j-1)$-th
     root (one can take function which is a product of $e^{-1/(t-a)}$), and the smoothness of the roots breaks only in points where the functions vanishes. As a consequence $\phi_i(f)$, $i=1,2$ is affiliated with $\A^+(I)$.
     This implies that $\phi_{j_1}(f_1)$ and $\phi_{j_2}(f_2)$ strongly commute
     if $f_1$ and $f_2$ have disjoint supports.

     We now come to the second part of the lemma. Arguing again as in \cite[Lemma 6.5]{CKLW18} and the proof of \cite[Proposition 6.6]{CKLW18} we can assume that $f$ is a real valued function.  As $\psi = J, \underline L$ satisfies the linear energy bound (see e.g. \cite{CKLW18}), we know that $\psi(f)$ is self-adjoint. 
    Now, for any $\Psi \in C^\infty(L_0$) let $\Psi(t) = Y(a,g)e^{it\psi(f)}\Psi$.  It follows by \cite[Proposition 2.2]{Toledano-Laredo99-1} and \cite[Corollary 2.2]{Toledano-Laredo99-1} together with the polynomial energy bounds for $Y(a,z)$ that  
     \begin{align*}
     \frac{d}{dt} \Psi(t) &=  iY(a,g)\psi(f)e^{it\psi(f)}\Psi = i\psi(f)Y(a,g)e^{it\psi(f)}\Psi  \\
     &=i\psi(f)\Psi(t) \,.
     \end{align*} 
     Hence, by the uniqueness of the solution of the abstract Schroedinger equation we have 
     $Y(a,g)e^{it\psi(f)}\Psi = e^{it\psi(f)}Y(a,g)\Psi$ for all $\Psi \in C^\infty(L_0)$ and all $t\in \mathbb{R}$ so that  $Y(a,g)$ commutes with $e^{it\psi(f)}$ and hence $Y(a,g)$ and $\psi(f)$ strongly commute. 
        
    \end{proof}

 The following is immediate from Lemma \ref{lm:localfield} and \cite[Theorem 8.1]{CKLW18}
  and generalizes \cite[Theorem 8.3]{CKLW18}.
    \begin{theorem}\label{th:localnet}
  Assume that a simple unitary VOA $V$ is generated by hermitian primary fields $\{\phi^{(j)}(z)\}$
 with conformal dimension $d_j$
 satisfying either a optimal global energy bound or a local energy bound with $s_j = d_j-1$,
 Virasoro fields and currents.
 Then it is strongly local.
    \end{theorem}    
    
    Although fields have been analyzed case-by-case until now,
    we know no primary field which violates the optimal energy bounds and expect the following.
    
\begin{conjecture}
\label{optimalprimary}
 Let $V$ be a simple unitary VOA. Then every primary vertex operator $Y(a,z), a \in V\setminus V_1$ satisfies optimal energy bounds. 
\end{conjecture}

\begin{remark} Any simple unitary vertex operator algebra $V$ is a direct sum of irreducible unitary positive energy representations of the Virasoro algebra. Hence, $V$ is generated by $L(z)$ and by Hermitian primary fields. Accordingly, if Conjecture \ref{optimalprimary} is true then 
$V$ is strongly local as a consequence of Theorem \ref{th:localnet} and this would imply  \cite[Conjecture 8.18]{CKLW18}.
\end{remark}

    \section{The \texorpdfstring{$\cW_3$}{W3}-algebra as a conformal net}
    \label{sec:W}
    We present concrete examples of primary fields that can be shown to be strongly local
    by Theorem \ref{th:localnet}: the $\cW_3$-algebra with $c>2$, whose definition we recall below.

    \subsection{The \texorpdfstring{$\cW_3$}{W3}-algebra and its lowest weight representations}\label{w3}
    The $\cW_3$-algebra is defined, for $c\neq -\frac{22}5$,
    as a pair of families of operators $\{L_n\}_{n\in \ZZ}, \{W_n\}_{n\in \ZZ}$
    on a $\CC$-linear space $V$ such that
    \begin{align*}
    \text{for each } \Psi \in V, \quad L_n\Psi = W_n\Psi = 0 \quad \text{ for sufficiently large } n
    \end{align*}
    ($n$ depends on $\Psi$) satisfying \cite{BS93, Artamonov16}:
    \begin{align}
     [L_m, L_n] &= (m-n)L_{m+n} + \frac{c}{12}m(m^2-1)\delta_{m+n,0}, \nonumber\\
     [L_m, W_n] &= (2m-n)W_{m+n}, \nonumber \\
     [W_m, W_n] &= \frac{c}{3\cdot 5!}(m^2-4)(m^2-1)m\delta_{m+n,0} \nonumber \\
                &\quad + b^2(m-n)\L_{m+n} + \left[\frac1{20}(m-n)(2m^2-mn+2n^2-8))\right]L_{m+n}, \label{eq:comm}
    \end{align}
    where $b^2 = \frac{16}{22+5c}$ and
    $\L_n = \sum_{k>-2} L_{n-k} L_k + \sum_{k\le-2} L_k L_{n-k} - \frac3{10}(n+2)(n+3)L_n$.
    We call such a pair a representation of the $\cW_3$-algebra.
    
    For a given triple $(c,h,w) \in \CC^3, c\neq -\frac{22}5$, a representation of
    the $\cW_3$-algebra is said to be a \textbf{lowest weight representation} if there is a cyclic vector
    $\Omega = \Omega_{c,h,w}$ satisfying
    \begin{align*}
     L_n\Omega &= W_n\Omega = 0 \text{ for } n > 0, \nonumber\\
     L_0\Omega &= h\Omega, \\
     W_0\Omega &= w\Omega.\nonumber
    \end{align*}
    Lowest weight representations exist for any such $(c,h,w)$
    (which is nontrivial because the $\cW_3$-algebra is not a Lie algebra, hence the standard
    construction through the universal enveloping algebra does not work, see e.g.\! \cite{CTW22}). 
    These representations can be in general reducible but, for any triple $(c,h,w) \in \CC^3, c\neq -\frac{22}5$ there 
    exists, up to isomorphism, a unique irreducible lowest weight representation.    
    
    We showed in \cite{CTW22} that, for any $c \ge 2, h=w=0$,
    there is a positive-definite sesquilinear invariant form $\<\cdot, \cdot\>$
    (an invariant scalar product) on
    the irreducible lowest weight representation (the \textbf{vacuum representation}) $V_{c,0,0}^{\cW_3}$:
    it holds that $(L_n)^\dagger = L_{-n}, (W_n)^\dagger = W_{-n}$ with respect to this form.

    \subsection{An energy bound for \texorpdfstring{$:L^2:$}{L2}}
    \label{sec:T^2}
    From the commutation relations \eqref{eq:comm}, the components $\{W_n\}$ have the conformal dimension $d=3$,
    hence we set $W(z) = \sum_n W_n z^{-n-3}$.
    The degree of the optimal energy bound for $W(z)$ 
    in Proposition \ref{pr:globaltolocal} is $s=d-1=2$. 
    To apply Proposition \ref{pr:improving} to the $\cW_3$-algebra in order to obtain
    the energy bound with $s = 2$, we need to estimate the right-hand side of
    \eqref{eq:comm} which contains $:L^2:$ and
    to have a bound with $s = 3$.
    Again from the commutation relation \eqref{eq:comm}, it is rather straightforward
    to obtain a bound with $s = 4$, but this is not enough to apply Theorem \ref{th:localnet}.
    With more computations using concrete representations, we can obtain $s = 3$, as we show below.
    More precisely, we will show the following bound.
     \[
      0 \le :L^2:_0 \le 5(c+1)(\1+L_0^3).
     \]

    \paragraph{Algebraic properties of $:L^2:$.}
    Let $\{L_n\}$ be a positive-energy unitary representation of the Virasoro algebra with $c > 0$ on $V$.
    The Fourier modes of the Wick square of $L(z) = \sum_n L_n z^{-n-2}$ is given by
    \[
     :L^2:_n = \sum_{k > -2} L_{n-k} L_k + \sum_{k \le -2} L_k L_{n-k}.
    \]
    By direct computations, we see that
    \begin{align*}
     [L_0, :L^2:_n] &= -n:L^2:_n \\
     [L_1, :L^2:_n] &= (3-n):L^2:_{n+1}
    \end{align*}
    In other words, $:L(z)^2:$ is quasi-primary and has conformal dimension $4$. From the Virasoro algebra relations,
    it follows that
    \begin{align}\label{eq:t20}
     :L^2:_0 = 2L_0 + L_0^2 + 2 \sum_{k\ge 1} L_{-k}L_k,
    \end{align}
    therefore, $:L^2:_0$ is a nonnegative operator which leaves invariant each eigenspace of $L_0$.
    
    \paragraph{An estimate in the oscillator representation, case $c=1, h=0$.}
    In order to derive a bound on $:L^2:_0$ in terms of $L_0$, we work in the vacuum representation $V_0^\uone$
    of the $\uone$-current algebra.
    The $\uone$-current  (or Heisenberg) algebra is an infinite-dimensional Lie algebra generated
    by $\{J_n\}_{n\in\ZZ}$ and a central element $Z$ with commutation relations
    \begin{align*}
     [J_m, J_n] = m\delta_{m+n,0} Z.
    \end{align*}
    It admits a vacuum representation:
    there is a representation $V^{\uone}$ where there is
    a cyclic vector $\Omega$ such that 
    \[
     J_m\Omega =0 \text{ for all } m \ge 0.
    \]
    Here vectors of the form
    \[
     J_{-n_1}\cdots J_{-n_k}\Omega,
    \]
    where $1 \le n_1 \le \ldots \le n_k$, form a basis.
    There exists a positive-definite sesquilinear form $\<\cdot,\cdot\>$ on
    $V_0^\uone$ with normalization
    $\<\Omega,\Omega\>=1$ such that $\<J_n \Psi, \Phi\> = \<\Psi, J_{-n}\Phi\>$,
    see e.g.\! \cite{KR87}.

    A representation of the Virasoro algebra with $c=1$ is given by
    \begin{align}\label{eq:sugawara}
     L_n = \frac12 :J^2:_n = \frac12\left(\sum_{k>-1} J_{n-k} J_{k} + \sum_{k\le -1} J_k J_{n-k}\right).
    \end{align}
    Let us denote $(V_0^\uone)_n = \ker (L_0 - n\1)$, then $(V_0^\uone)_n$ are pairwise orthogonal and finite-dimensional, and they span $V_0^\uone$.
    On each vector $\Psi \in V_0^\uone$, the sum
    \[
     J_- = \sum_{n > 0} J_n
    \]
    becomes finite, hence $J_-$ is a linear operator on $V_0^\uone$. On the other hand,
    \[
     J_+ = \sum_{n < 0} J_n
    \]
    cannot be defined as a linear operator. Nevertheless, it makes sense as a sesquilinear form:
    \[
     \<\Psi_1, J_+\Psi_2\> := \sum_{n < 0}\<\Psi_1, J_n\Psi_2\> = \<J_-\Psi_1, \Psi_2\>.
    \]
    In this sense, $J_+$ can be seen as the formal adjoint of $J_-$, although it is not an operator.
    Similarly, $J_-J_-$ makes sense as an operator, and $J_+J_+, J_+J_-$ can be defined as sesquilinear forms
    (but $J_-J_+$ cannot).
    We can also define the infinite sum $\sum_n L_n$ as a sesquilinear form, namely
    $\<\Psi_1, \sum_n L_n \Psi_2\> := \<\sum_{n>0} L_n\Psi_1, \Psi_2\> + \<\Psi_1, \sum_{n \ge 0}L_n \Psi_2\>$.

    With this understanding, the following holds as sesquilinear forms
    by straightforward computations from the Sugawara formula \eqref{eq:sugawara}
    (for example, by counting how many times a product $J_k J_\ell$ appears in both sides)
    and the fact $J_0 = 0$ in $V_0^\uone$:
    \begin{align}\label{eq:sumLn}
     \sum_n L_n = \frac12(J_+J_+ + J_-J_-) + J_+J_-.
    \end{align}
    
    Next, let us observe that if $L_0 \Psi = n\Psi$, then $\|J_- \Psi\| = \sqrt n \|\Psi\|$.
    Indeed, as the vectors $J_k \Psi \in (V_0^\uone)_{n-k}$ are pairwise orthogonal,
    \[
     \|J_-\Psi\|^2 = \sum_{k > 0} \|J_k \Psi\|^2 = \sum_{k > 0} \<\Psi, J_{-k}J_k\Psi\> = \<\Psi, L_0\Psi\> = n\|\Psi\|^2.
    \]
    From this it follows that
    \[
    \text{if }\quad \Psi \in \bigoplus_{k=0}^n (V_0^\uone)_k, \quad\text{ then } \quad
    \|J_-\Psi\| \le n \|\Psi\|.
    \]
    Indeed, by writing $\Psi = \sum_{k=0}^n \Psi_n, \Psi_n \in (V_0^\uone)_n$, we have
    \[
     \|J_-\Psi\|^2 = \left\|\sum_{k=0}^n J_-\Psi_k\right\|^2 \le n\sum_{k=0}^n \|J_-\Psi_k\|^2 = n\sum_{k=1}^n k\|\Psi\|^2 \le n^2 \|\Psi\|^2.
    \]
    
    The first nontrivial estimate of degree $3$ can be obtained in this representation.
    \begin{lemma}\label{lm:Lk}
     For $\{L_k\}$ in $V_0^\uone$ defined above and for any $\Psi \in (V_0^\uone)_n$, it holds that
     \[
      0 \le \left\<\Psi, \sum_{k\ge 0} L_{-k}L_k\Psi\right\> \le \frac94 n^3
     \]
    \end{lemma}
    \begin{proof}
     The vectors $L_k\Psi$ are pairwise orthogonal for $k\in \ZZ$, and they are contained in
     \[
      (V_0^\uone)_{\le n} := \bigoplus_{j=0}^n (V_0^\uone)_j.
     \]
     If $k < 0$, then $L_k\Psi$ is orthogonal to $(V_0^\uone)_{\le n}$. Therefore,
     \[
      \left\<\Psi, \sum_{k\ge 0} L_{-k}L_k\Psi\right\> = \left\| \sum_{k\ge 0} L_k\Psi\right\|^2
      = \underset{\Phi\in (V_0^\uone)_{\le n}, \|\Phi\| = 1}\max\left\{\left|\left\<\Phi, \sum_{k\in\ZZ} L_k\Psi\right\>\right|^2\right\}
     \]
     By substituting $\sum_{k\in\ZZ} L_k$ by the right-hand side of \eqref{eq:sumLn} and by the fact that
     $J_-J_- (V_0^\uone)_{\le n} \subset (V_0^\uone)_{< n}$, and hence $J_-J_- (V_0^\uone)_{\le n}$ is
     orthogonal to $\Psi$, we obtain
     \begin{align*}
      \left\<\Psi, \sum_{k\ge 0} L_{-k}L_k\Psi\right\> 
      &= \underset{\Phi\in (V_0^\uone)_{\le n}, \|\Phi\| = 1}
      \max\left\{\left|\left\<\Phi, \frac12 J_-J_- \Psi\right\> + \left\<J_-\Phi, J_- \Psi\right\>\right|^2\right\} \\
      &\le \left(\frac12n\sqrt n + n\sqrt n\right)^2 \|\Psi\|^2 \le \frac94 n^3 \|\Psi\|^2.
     \end{align*}
    \end{proof}
    
    \paragraph{An estimate in the oscillator representation, case $c>1$.}
    We need to extend the above results to $c > 1$. For this purpose, we take the representation
    of the Virasoro algebra on $V_0^\uone$ with $c>1$ \cite[Section 3.4]{KR87}. Its Fourier coefficients are
    \begin{align}\label{eq:virasoro-nonunitary}
     L^{\k, \eta}_n := \begin{cases}
                        L_0 + \frac12 (\k^2 + \eta^2)\1, \quad \text{if } n=0, \\
                        L_n + (\k + in\eta)J_n, \quad \text{if } n \neq 0.
                       \end{cases}
    \end{align}
    Although these representations are \textit{not unitary} for $\k,\eta \notin \RR$,
    we can still use them in order to obtain energy bounds.
    
    \begin{lemma}
     Let $\k \in \RR \setminus \{0\}$ and $\eta \in \RR \cup i\RR$ such that
     $\k^2 + \eta^2 > 0$.
     Then the representation \eqref{eq:virasoro-nonunitary} is isomorphic to the Verma module
     with $c = 1+\k^2 > 1$ and $h = \frac{\k^2 + \eta^2}2$,
     and in particular, is irreducible.
    \end{lemma}
    \begin{proof}
     The subspaces $(V_0^\uone)_n$ remain eigenspaces of $L^{\k,\eta}_0$, but with
     eigenvalues $n+h$. In particular, the vacuum vector $\Omega$ for the $\uone$-current is the lowest weight vector for this representation.
     For the values $c > 1, h > 0$, the Kac determinants of the Virasoro algebra are positive, hence
     the Verma module is irreducible \cite[Proposition 8.2(b)]{KR87}.
     Therefore, the restriction of $\{L^{\k, \eta}_n\}$ to
     the subspace $M$ generated from $\Omega$ by  $\{L^{\k, \eta}_n\}$ is equivalent to the Verma module.
     
     It remains to show that $M = V_0^\uone$. This follows because for each $n$,
     $\dim (V_0^\uone)_n = p(n)$, where $p(n)$ is the number of partitions of $n$,
     while $(V_0^\uone)_n \supset \ker (L^{\k, \eta}_0|_M - (n+h))$ and $\dim(\ker (L^{\k, \eta}_0|_M - (n+h))) = p(n)$
     as well. 
    \end{proof}
    
    \begin{proposition}\label{pr:t2bound}
     Let $\{L_n\}$ be the irreducible lowest weight representation of the Virasoro algebra with
     $c>1, h>0$. Then for every $\Psi \in V$,
     \[
      \text{if }\quad L_0\Psi = (n+h)\Psi, \quad \sum_{k\ge 0} L_{-k}L_k \Psi = \l\Psi,
      \quad \text{then }\quad |\l| \le 5(c+1)(n+h)^3.
     \]
    \end{proposition}
    \begin{proof}
     As the statement is purely algebraic and does not depend on the scalar product, we can use
     the representation $\{L^{\k, \eta}_n\}$ to prove the thesis.
     The scalar product $\<\cdot,\cdot\>$ below is the one of the $\uone$-current algebra,
     and $(L^{\k,\eta}_n)^\dagger \neq L^{\k,\eta}_{-n}$ with respect to $\<\cdot, \cdot\>$ in general.
     However, $(L^{\k,\eta}_0)^\dagger = L^{\k,\eta}_0$ and hence is diagonalizable on each
     $\ker (L^{\k, \eta}_0 - (n+h)) = (V_0^\uone)_n$. Therefore, it is enough to work on
     eigenvectors of $L^{\k, \eta}_0$ in $(V_0^\uone)_n$. 
     
     Let $\Psi \in \ker (L^{\k, \eta}_0 - (n+h)) = (V_0^\uone)_n$
     and assume that $\sum_{k\ge 0} L^{\k, \eta}_{-k}L^{\k, \eta}_k \Psi = \l\Psi$.
     Set
     \[
      L_-^0 := \sum_{k \ge 0} L^{0,0}_k, \quad J_-' := \sum_{k > 0} kJ_k.
     \]
     which are operators on $V_0^\uone$.
     By noting that $(L_k^{0,0})^\dagger = (L_k)^\dagger = L_{-k} = L_{-k}^{0,0}$,
     it is straightforward to show that
     \[
      \l\|\Psi\|^2 = \left\< \Psi, \sum_{k\ge 0} L^{\k, \eta}_{-k}L^{\k, \eta}_k \Psi\right\>
      = \<(L_-^0 + \k J_- + i\bar \eta J_-')\Psi, (L_-^0 + \k J_- + i\eta J_-')\Psi\>,
     \]
     because, upon expanding the last expressions, terms with different $k$ are orthogonal.
     We may assume that $\Psi \neq 0$ and hence
     \[
      |\l| \le \frac{1}{\|\Psi\|^2} (\|L_-^0\Psi\| + \|\k J_-\Psi\| + \|\eta J_-'\Psi\|)^2.
     \]
     We know that $\|J_-\Psi\| = \sqrt n \|\Psi\|$.
     The vector $L^0_-\Psi$ is a sum of eigenvectors of $L_0$ which are orthogonal to each other,
     hence $\|L^0_-\Psi\|^2 = \|\sum_{k\ge 0} L_{-k}^{0,0}L_k^{0,0}\Psi\|^2 \le \frac94 n^3$ by Lemma \ref{lm:Lk}, or
     $\|L^0_-\Psi\| \le \frac32 n^\frac32\|\Psi\|$. As for $J_-'\Psi$,  we have
     \[
      \|J_-'\Psi\|^2 = \sum_{k=1}^n k^2 \|J_k\Psi\|^2 \le n^2 \sum_{k=1}^n \|J_k\Psi\|^2
      = n^2 \sum_{k=1}^n \<\Psi, J_{-k}J_k\Psi\>
      = n^2 \<\Psi, L_0^{0,0}\Psi\>
      = n^3\|\Psi\|^2.
     \]
     By assumption $c > 1$.
     In addition, if $\eta^2 < 0$, then $|\eta|^2 < \k^2 = \frac{c-1}{12}$
     and if $\eta^2 > 0$, then $|\eta|^2 \le \eta^2+\k^2 = 2h$.
     Altogether,
     \begin{align*}
      |\l| &\le \left(\frac32 n^\frac32 + |\k|n^\frac12 + |\eta| n^\frac32\right)^2 
      \le \left(\frac32 n^\frac32 + \left(\sqrt{\frac{c-1}{12}}+\sqrt{2h}\right)n^\frac12
      + \sqrt{\frac{c-1}{12}}n^\frac32\right)^2 \\
      &\le 5(c+1)(n+h)^3
     \end{align*}
     as claimed.
    \end{proof}

    \begin{corollary}\label{cr:t2bound}
     Let $\{L_n\}$ be a direct sum $V$ of irreducible lowest weight representations of the Virasoro algebra with
     $c\ge 1, h\ge 0$. Then, for every $\Psi \in V$,
     \[
      0 \le \<\Psi, :L^2:_0\Psi\> \le \<\Psi, 11(c+1)(\1+L_0^3)\Psi\>.
     \]
    \end{corollary}
    \begin{proof}
     It is enough to prove this for irreducible components. Furthermore, we first assume that
     $c > 1, h > 0$.
     
     By \eqref{eq:t20}, $:L^2:_0$ is a nonnegative operator.
     Moreover, as $\max\{2x-x^2: x\in \RR\} = 1$, we find
     (note that the sum below runs from $k=0$ while it runs from $k=1$ in \eqref{eq:t20})
     \[
      0 \le :L^2:_0 \le \1 + 2\sum_{k\ge 0} L_{-k}L_k.
     \]
     The sum $\sum_{k\ge 0} L_{-k}L_k$ leaves invariant each eigenspace of $L_0$,
     which are finite dimensional, therefore, $\sum_{k\ge 0} L_{-k}L_k$ is a direct sum
     of positive matrices. Proposition \ref{pr:t2bound} gives an estimate
     on a simultaneous eigenvector $\Psi$ of $L_0$ and $\sum_{k\ge 0} L_{-k}L_k$, hence
     we have $0 \le \sum_{k\ge 0} L_{-k}L_k \le 11(c+1)L_0^3$.
     This establishes the claim for $c> 1, h>0$.
     
     As for $c \ge 1, h \ge 0$, we take the Verma module $V_{c,h}$.
     These Verma modules can be identified as the space generated by $L_{-n_1}\cdots L_{-n_k}\Omega$
     even for different values of $c,h$.
     For any fixed $\Psi \in V_{c,h}$,
     $\<\Psi, (11(c+1) (\1+L_0^3)-:L^2:_0 )\Psi\>$ is a polynomial in $c,h$
     (we consider $\Psi$ as a fixed vector in the above single vector space, while changing $c,h$).
     As we have proven the non-negativity of this polynomial for $c>1, h>0$,
     the claim for $c\ge1, h\ge0$ follows by continuity.
    \end{proof}

    \subsection{Strong locality of the \texorpdfstring{$\cW_3$}{W3}-algebra}
    \label{sec:strongW}
    
    Let $W(z), L(z)$ be a lowest weight representation of the $\cW_3$-algebra with
    $c \ge 1, h = 0, w=0$ and assume that it is unitary.
    We fix $k \in \NN_+$. From the commutation relations \eqref{eq:comm},
    $[W_{-k}, W_k]$ can be bounded by $C(L_0+\1)^3$
    ($C$ depends on $k$, but this is not important here):
    the worst term is $\Lambda_0 = :L^2:_0 - \frac95 L_0$, whose bound has been obtained in Corollary \ref{cr:t2bound}
    (here the condition $c\ge 1$ is necessary),
    and the rest is a linear combination of $L_0$ and $c$.
    From Propositions \ref{pr:component}, \ref{pr:canonicalbound} and \ref{pr:improving} we infer that
    $W_k$ is bounded by $C'(L_0+\1)^2$.
    As the conformal dimension of $W(z)$ is $3$, this is the optimal global energy bound with $d=3, s=2$.
    By Proposition \ref{pr:globaltolocal}, we obtain a local energy bound.

    As we observed in \cite{CTW22}, they generate a simple unitary VOA
    $\cW_{3,c}$ (see \cite[Section 4.3]{CKLW18} for a review of unitary VOA and our conventions, and remarks before Lemma \ref{lm:localfield}).
    Now Theorem \ref{th:localnet} applies directly to them:
    \begin{corollary}\label{cr:w3net}
     For $c \ge 1$ such that the VOA $\cW_{3,c}$ is unitary,
     it is strongly local.
    \end{corollary}
    While strong locality for discrete series ($c < 2$) is known by the coset construction \cite{ACL19}, \cite[Section 5.1]{Tener19-2}, and $c=2$ is a subnet of the tensor product of the $\uone$-current net,
    in the continuous region $c > 2$ strong locality is settled for the first time by local energy bound. Note that the that our proof of strong locality works also for the $c$ in the discrete series with $c\geq 1$ and, except from the trivial case $c=0$, there is only one value of $c$ in the discrete series with $c<1$ namely $c=2(1- \frac35)=0.8$. In the latter case our proof does not work although strong locality follows from the coset realization.   
    
    We do not know whether $W(g)$ is essentially self-adjoint on $C^\infty(L_0)$
    for $g$, not of the form $f^{d-1} = f^2$ where $f$ is a nonnegative function.

    \subsection{Some properties of the \texorpdfstring{$\cW_3$}{W3}-nets with \texorpdfstring{$c > 2$}{c>2}}
    \label{some}
    \paragraph{Embedding into the free field net on $S^1 \setminus \{-1\}$.}
    
    In \cite[Section 3.2 (19)]{CTW22} we constructed a representation of the $\cW_3$-algebra
    on $V_0^\uone\otimes V_0^\uone$. The fields $\tilde L(z), \tilde W(z)$ are linear combinations of certain normal products of
    two currents and a function $\rho(z) = -i\frac{z-1}{z+1}$.
    The only singularity of $\rho(z)$ is at $z=-1$, therefore, it is possible to smear $\tilde L(z), \tilde W(z)$
    with functions whose supports do not contain $-1$, and obtain operators affiliated with
    the local algebras of the tensor product net $\A_{\uone}\otimes \A_{\uone}$
    of two $\uone$-current nets. They generate a translation-dilation covariant subnet $\mathcal{B}$ of $\A_{\uone}\otimes \A_{\uone}$.
    Let $P_\mathcal{B}$ be the orthogonal projection onto the subspace generated by $\mathcal{B}$,
    Then $P_\mathcal{B}$ commutes any element of the subnet $\mathcal{B}$, and in particular it commutes (strongly) with $\tilde L(f), \tilde W(g)$.
    Furthermore, if the supports of $f,g$ are disjoint, these operators commute strongly.
   
    
    Note that $P_\mathcal{B}\mathcal{H}_{\uone\times\uone}$ contains all the vectors generated from $\Omega$ by $\tilde L(f), \tilde W(g)$,
    where $-1 \notin \supp g, \supp f$. These vectors can be identified with those in the completion $\mathcal{H}_{\cW_3,c}$ of
    the vacuum module of the $\cW_3$-algebra, generated by $L(f), W(g)$ from the vacuum vector.
    By Corollary \ref{cr:w3net}, we have a conformal net $\A_{\cW_3, c}$ on $\mathcal{H}_{\cW_3,c}$ generated by $L(f), W(g)$,
    therefore, by the Reeh-Schlieder property, the subspace $\mathscr{D}_{\cW_3,c}$ generated from the vacuum by $L(f), W(g)$ with $-1\notin \supp f, \supp g$
    is dense in $\mathcal{H}_{\cW_3,c}$, and we can identify $\mathcal{H}_{\cW_3, c}$ as a closed subspace of
    $P_{\mathcal{B}}\mathcal{H}_{\uone\times \uone}$. Moreover,  $\mathscr{D}_{\cW_3,c}$ is a core for all the operators
    $L(f), W(g)$ with $-1\notin \supp f, \supp g$, e.g. by \cite[Lemma 7.1]{CKLW18}, cf. also \cite[Eq. (4.90)]{BS90}.            
     
     Let $P_{\cW_3,c}$ be the projection onto this space. 
    Let $g = \tilde g^2$ for some $\tilde g \ge 0$. Then, $W(g)$ is a self-adjoint operator on  $\mathcal{H}_{\cW_3,c}$  by
    Theorem \ref{th:strong}  and Section \ref{sec:strongW}.     
    Let us consider the restriction $\hat{W}(g)$ to  $\mathcal{H}_{\cW_3,c}$ of the operator $P_{\cW_3,c} \tilde W(g) P_{\cW_3,c}$, namely the operator on $\mathcal{H}_{\cW_3,c}$ with dense domain  $\mathcal{H}_{\cW_3,c} \cap \dom(\tilde{W}(g))$ acting as 
    $P_{\cW_3,c} \tilde{W}(g)$
    on this domain. $\hat{W}(g)$ is a symmetric operator extending the self-adjoint operator $W(g)$. 
    Hence, $\hat{W}(g)=W(g)$ and it follows that $P_{\cW_3,c} \tilde W(g) P_{\cW_3,c}$ is a self-adjoint operator on          
    $\mathcal{H}_{\uone\times \uone}$. Moreover, if $\Psi \in \mathcal{H}_{\cW_3,c} \cap \dom(\tilde{W}(g))$   then 
    there is a sequance $\Psi_n \in  \mathscr{D}_{\cW_3,c}$ such that $\Psi_n \to \Psi$ and 
    $\tilde{W}(g)\Psi_n =$ $P_{\cW_3,c} \tilde{W}(g)\Psi_n \to P_{\cW_3,c} \tilde{W}(g)\Psi $
    so that $\tilde{W}(g)\Psi = P_{\cW_3,c} \tilde{W}(g)\Psi$. As a consequence we have the operator equality  
    $\tilde{W}(g)P_{\cW_3,c} = P_{\cW_3,c} \tilde{W}(g)P_{\cW_3,c}$  and, the latter being self-adjoint, we can conclude that  $\tilde{W}(g)P_{\cW_3,c}  = ( \tilde{W}(g)P_{\cW_3,c}  )^*$. But, since $\tilde{W}(g)$ is symmetric,   
    $( \tilde{W}(g)P_{\cW_3,c}  )^* $,  is an extension of $P_{\cW_3,c}\tilde{W}(g)$ i.e. $P_{\cW_3,c}$ commutes with
    $\tilde{W}(g)$. Similarly, $\tilde{L}(f)$ and $P_{\cW_3,c}$ commute, c.f. \cite[Section 4]{BS90}. In the latter case the 
    self-adjointness of  $L(f)$ for a real $f$ follows from the linear energy bounds. Therefore,
    $P_{\cW_3,c} = P_{\mathcal{B}}$. In this way, the net $\mathcal{A}_{\cW_3,c}$ can be
    embedded in $\A_{\uone}\otimes \A_{\uone}$ as a translation-dilation covariant subnet.    
    
    \begin{remark} The above argument gives an an alternative proof of the strong locality of 
    ${\cW_3,c}$ for $c\geq 2$. Note however that the argument is only partially independent from Theorem \ref{th:strong} since the self-adjointness of $W(\tilde{g}^2)$ was used in a crucial way.      
      
    \end{remark}

    \paragraph{Failure of strong additivity.}
    Let $c > 2$.
    We can take localized Weyl operators, e.g.\! $e^{iJ(f)}\otimes \1$ in $\A_{\uone}\otimes \A_{\uone}$,
    where $J$ is the $\uone$-current.
    By the arguments of \cite[Section 4]{BS90}, the operator $P(e^{iJ(f)}\otimes \1)P$ belongs to
    the dual net of $\A_{\cW_3, c}$, but not $\A_{\cW_3, c}$ itself. This shows that
    $\A_{\cW_3, c}$ does not satisfy strong additivity, just as the Virasoro nets with $c > 1$.

    \paragraph{Infinite $\mu$-index.}
    While $\A_{\cW_3,c}$ does not satisfy strong additivity, it is conformally covariant,
    hence it satisfies the split property \cite[Theorem 5.4]{MTW18}.
    By \cite[Theorem 5.3]{LX04}, it cannot have finite $\mu$-index in the sense of \cite{KLM01}.

    \paragraph{Some locally normal representations.}
    Again, as with the Virasoro nets with $c > 1$, we can construct some locally normal representations
    using the automorphisms of the $\uone$-currents. Indeed, we have constructed some
    lowest weight representations which are manifestly unitary \cite[Section 3.4]{CTW22}.
    These representations are obtained by composing certain locally normal automorphisms
    of $\A_{\uone}\otimes \A_{\uone}$ (except $-1$) and a representations of the form \cite{BMT88}, cf. 
    \cite[Section 4]{BS90}.
    This establishes local normality of the lowest weight representations with
     \[
     h\geq \frac{c-2}{24},\;\; |w|\leq \sqrt{\frac{8}{198+45c}}
     \left(2h-\frac{c-2}{12}\right)^{\frac{3}{2}}\;\;\;\;(h,w\in\RR).
     \]
    For other representations, analytic continuation might be useful \cite{Weiner17}.

    \subsection*{Acknowledgements}
We would like to thank the referees for helpful comments.
SC and MW are supported in part by the ERC advanced grant 669240 QUEST ``Quantum Algebraic Structures and Models''.
SC is also supported by GNAMPA-INDAM.
YT was supported until February 2020 by Programma per giovani ricercatori, anno 2014 ``Rita Levi Montalcini''
of the Italian Ministry of Education, University and Research.
MW is supported also by the National Research,
Development and Innovation Office of Hungary (NRDI) via the research grant K124152, KH129601 and K132097
and the Bolyai J\'anos Fellowship of the Hungarian Academy of Sciences, the \'UNKP-19-4 New National Excellence Program
of the Ministry for Innovation and Technology.

SC and YT acknowledge the MIUR Excellence Department Project awarded to
the Department of Mathematics, University of Rome ``Tor Vergata'' CUP E83C18000100006 and the University of
Rome ``Tor Vergata'' funding scheme ``Beyond Borders'' CUP E84I19002200005.

We are also grateful to Mathematisches Forschungsinstitut Oberwolfach and Simons Center for Geometry and Physics,
where parts of this work have been done.

{\small
\newcommand{\etalchar}[1]{$^{#1}$}
\def\cprime{$'$} \def\polhk#1{\setbox0=\hbox{#1}{\ooalign{\hidewidth
  \lower1.5ex\hbox{`}\hidewidth\crcr\unhbox0}}} \def\cprime{$'$}

}

\end{document}